\documentclass[11pt]{amsart}
 
\theoremstyle{plain}
\newtheorem{thm}{Theorem}[section]
\newtheorem{theorem}[thm]{Theorem}

\newtheorem{lemma}[thm]{Lemma}
\newtheorem{corollary}[thm]{Corollary}

\theoremstyle{definition}

\newtheorem{definition}[thm]{Definition}

\numberwithin{equation}{section}

\title {Time and Space separation in General Relativity}
 \author{Tuyen Trung Truong}
        \email{truongt@indiana.edu}
 
\begin{document}

\maketitle

\begin{abstract}
Let $(M,g)$ be a spacetime. That is, $M$ is a real manifold of dimension $4$ equipped with a Lorentzian metric $g$. We show that any separation of time and space in $M$ is equivalent to introducing a (non-smooth) Riemann metric $h$. If $h$ is smooth, it induces a smooth line bundle $T_p\rightarrow M$, whose any fiber is generated by a time-like vector, called the time bundle. Whether $(M,g,h)$ is time orientable or not corresponds to whether this line bundle is trivial or not. As well-known, the last condition is characterized by the first Stiefel-Whitney class $w^1(T_p)\in H^1(M,\mathbb{Z}/2)$. We then define a partial time orientation of $M$ as a section of the line bundle $T\rightarrow M$. As applications, we discuss time and space differentiations on $M$.
\end{abstract}
\section{Introduction}
The relation between space and time has been discussed since the early time of human history, and in many areas: science, philosophy, literature,... For references, see for  example \cite{hawking, penrose, smolin, wald}.

In Newtonian mechanics, space and time are absolute notions and are separable. The relativity theory of Einstein revolutionized our understanding of space and time. In special relativity, Einstein showed that time and space are interrelated. However, we can still consider the time and space coordinates separately. In general relativity, when gravity is incorporated, then space and time coordinates have the same roles and we can not tell them apart. 

It is, however, for applications, that we want to redo the process above. That is, we want to, given a spacetime in general relativity, separate the space and time coordinates. The common approaches on this are the $3+1$ and $1+3$ separations, these corresponding to slicing the spacetime by hypersurfaces or by curves (see \cite{gourgoulhon} and references therein). Most recently, Smolin \cite{smolin} has a completely new perception on the nature of time.     

Our approach in this paper is different from the above. The main object of our paper is to give a simple separation of time and space, by employing a Riemannian metric. Given a Riemannian metric, we produce a smooth real line bundle $T\rightarrow M$, called the time bundle. This time bundle is the separation of time and space we seek for. Then we define a partial orientation of time as a smooth section of this time bundle. The question of whether $M$ is time orientable is reduced to whether whether $T$ is a trivial bundle, or equivalently, whether the first Stiefel-Whitney class of the time bundle is zero. We then gives applications of this  separation to time and space differentiations.     

\section{Main results} 

Let $(M,g)$ be a spacetime. That is, $M$ is a real manifold of dimension $4$ equipped with a Lorentzian metric $g$. Here we use the convention that the metric has signature $(-1,1,1,1)$. For each $p\in M$, we let $V_p$ denote the tangent space of $M$ at $p$. Then, (see \cite{wald}), the null cone $N_p$ and the light cone $L_p$ are defined as follows:
\begin{eqnarray*}
N_p&=&\{v\in V_p:~g_p(v,v)=0\},\\
L_p&=&\{v\in V_p:~g_p(v,v)\geq 0\}. 
\end{eqnarray*}
A vector $v\in V_p$ is called time-like if $g_p(v,v)<0$, null if $g_p(v,v)=0$, and space-like if $g_p(v,v)>0$. For each $v\in V_p$, we define $v^{\perp}=\{w\in V_p:~g(v,w)=0\}$.

The following simple result follows from Sylvester's law
\begin{lemma}
If $v\in V_p$ is time-like, then the restriction of $g_p$ to $v^{\perp}$ is positive definite. 
\label{LemmaOrthogonalSpace}\end{lemma}

In general relativity, time and space can not be separated.  In the remaining of this section, we show that to separate them is equivalent to introducing a (non-smooth) Riemann metric $h$. By separation of space and time, we mean that for each $p\in M$ there is given a time line $T_p$. Here is the precise definition. 
\begin{definition}
A separation of time and space is a (non-smooth) assignment $p\mapsto v_p\in V_p$, where $v_p$ is a time-like vector.  
\label{DefinitionSeparationTimeSpace}\end{definition}

Any separation of time and space is equivalent to introducing a (non-smooth) Riemannian metric, as shown in the below. 
\begin{theorem}
Any separation of time and space induces a (non-smooth) Riemannian metric $h_p$ on $V_p$, for any $p$.

Conversely, given a (non-smooth) Riemannian metric $h_p$ on $V_p$, then we have a separation of time and space.  
\label{TheoremSeparation}\end{theorem}
\begin{proof}
First, assume that we have a separation of time and space. We now construct a Riemannian metric $h_p$. Fix $p\in M$. Since $v_p$ is time-like, by Lemma \ref{LemmaOrthogonalSpace} the restriction of $g$ to the orthogonal space $v_p^{\perp}$ is positively definite. Any $v\in V_p$ can be written uniquely as $v=\lambda v_p+w$, where $\lambda \in \mathbb{R}$ and $w\in v_p^{\perp}$. Then, we define
\begin{eqnarray*}
h_p(v,v)=-\lambda ^2g_p(v_p,v_p)+g_p(w,w).
\end{eqnarray*}

Last, assume that we have a Riemannian metric $h_p$ on $V_p$. Then $h_p$ gives an identification of $V_p$ with $V_p^*$, under which $g_p$ defines a linear map $V_p\rightarrow V_p$. More precisely, for any $v\in V_p$, then $g_p(v)\in V_p$ is such that  for any $w\in V_p$
\begin{eqnarray*}
h_p(g_p(v),w)=g_p(v,w).
\end{eqnarray*}

Since both $h_p$ and $g_p$ are metrics (thus symmetric), the linear map $g_p:V_p\rightarrow V_p$ is symmetric. Therefore it is diagonalizable. Since the map $g_p$ has only one negative eigenvalue $\lambda _p$, there is then only one (upto multiplicative constants) non-zero eigenvector $v_p$ with eigenvalue $\lambda _p$. This vector $v_p$ is time-like since by definition
\begin{eqnarray*}
g_p(v_p,v_p)=h_p(g_p(v_p),v_p)=h_p(\lambda _pv_p,v_p)=\lambda _ph_p(v_p,v_p)<0.
\end{eqnarray*}
\end{proof}
   
On any smooth real manifold $M$, there exist smooth Riemannian metrics. In this case, we have the following result. 
\begin{theorem}
Let $h$ be a smooth Riemanninan metric on $(M,g)$. Then there is associated a smooth line bundle $T\rightarrow M$, such that for any $p\in M$ the fiber $T_p$ is generated by a time-like vector $v_p$. 
\label{TheoremSmoothSeparation}\end{theorem}
\begin{proof}
If $h$ is smooth, then locally the assignment $p\mapsto v_p$ in Theorem \ref{TheoremSeparation} can be made to be  smooth. Therefore, if we assign $T_p\subset V_p$ be the line in generated by $v_p$, we have a smooth line bundle $T\rightarrow M$. 
\end{proof} 

We call the line bundle $T\rightarrow M$ the time bundle. We say that the bundle $T\rightarrow M$ is time orientable if it is trivial. In this case, we can define globally on the whole of $M$, the notions of "past" and "future". 

Let $w_j(.)\in H^j(M,\mathbb{Z}/2)$ denote the $j$-th Stiefel-Whitney class of a vector bundle on $M$. Then, it is well-known that the smooth line bundle $T\rightarrow M$ is orientable iff its first Stiefel-Whitney class $w_1(T)$ is zero. 

We summarize this discussion as follows
\begin{theorem}
A spacetime $(M,g)$ is time-orientable iff there is a smooth Riemannian metric $h$ on $M$, such that the induced time line bundle $T\rightarrow M$ is trivial, or  equivalently iff the first Stiefel-Whitney class $w_1(T)$ is zero.   
\label{TheoremTimeOrientable}\end{theorem}
As a consequence, we obtain the following result. 
\begin{corollary}
Let $w_j(M)$ denote the $j$-th Stiefel-Whitney class of the tangent bundle of $M$. If $\omega _1(M)+\omega _3(M)\not= \omega _2(M)+\omega _4(M)$, then the spacetime $(M,g)$ is not time-orientable. 
\label{Corollary1}\end{corollary}
The above corollary is quite abstract and does not use the metric $g$. It is probably that, with the use of the metrics $g$ and $h$ together with the cohomology group of $M$, we may be able to say more on whether the spacetime $(M,g)$ is time orientable or not.

Whether or not the time bundle $T\rightarrow M$ is orientable, we can define a partial time orientation of $M$ as follows. 
\begin{definition}
We call a partial time orientation of $M$ a smooth section $s:T\rightarrow M$. If $s(p)\not= 0$, we say that $s(p)$ is the  future direction to an observer at $p$. If $s(p)=0$, we say that the partial time orientation $s$ can not distinguish between past and future at the point $p$. 
\label{DefinitionTimeOrientation}\end{definition}

We conclude this paper with an application on space and time differentiation. Having the time line bundle $T\rightarrow M$, we also have the space bundle $S\rightarrow M$, defined as the quotient bundle of $T$ in the tangent bundle of $M$.   

We let $\nabla ^g$ be the unique connection on $M$ which is compatible with $g$ (see \cite{wald}). Given $F$ a tensor on $M$, we say a time differentiation of $F$ an object
\begin{eqnarray*}
\nabla ^g_sF,
\end{eqnarray*}      
where $s:M\rightarrow T$ is a smooth section of the time bundle $T$. Similarly, if $s:M\rightarrow S$ is a section of the space bundle $S\rightarrow M$, we say $\nabla ^g_sF$  a space differentiation.

\end{document}